\documentclass[letterpaper]{article}
\usepackage{url}
\usepackage{amsmath}
\usepackage{amsthm}
\usepackage{color,graphicx,fancybox,texdraw}
\usepackage{amssymb}
\usepackage[all]{xy}
\usepackage{multicol}
\usepackage{subcaption}
\usepackage{color, colortbl}
\usepackage{makecell}
\usepackage{subfig}
\usepackage{slashbox}
\usepackage{threeparttable}

\theoremstyle{plain}
\newtheorem{lemma}{Lemma}
\newtheorem{theorem}[lemma]{Theorem}
\newtheorem{proposition}[lemma]{Proposition}

\theoremstyle{definition}
\newtheorem{definition}[lemma]{Definition}

\newtheorem{example}[lemma]{Example}

\numberwithin{equation}{section}
\numberwithin{lemma}{section}

\newcommand{\id}[1]{\langle#1\rangle}

\newcommand{\Z}{\mathbb{Z}}

\newcommand{\F}{\mathbb{F}}
\definecolor{Gray}{gray}{0.9}

\title{Duality Preserving Gray Maps for Codes over Rings}
\author{Steve Szabo \footnote{e-mail: steve.szabo@eku.edu}\\ Department of Mathematics and Statistics\\ Eastern Kentucky University\\ Richmond, KY 40475\\\\
Felix Ulmer\footnote{e-mail: felix.ulmer@univ-rennes1.fr}\\IRMAR, CNRS, UMR 6625\\ Universit\'e de Rennes 1\\ Universit\'e europ\'eenne de Bretagne\\
Campus de Beaulieu, F-35042 Rennes}
\date{}

\begin{document}
\maketitle

\begin{abstract}
Given a finite ring $A$ which is a free left module over a subring $R$ of $A$, two types of $R$-bases, pseudo-self-dual bases (similar to  trace orthogonal bases) and symmetric bases, are defined which in turn are used to define duality preserving maps from codes over $A$ to codes over $R$. Both types of bases are generalizations of similar concepts for fields. Many illustrative examples are given to shed light on the advantages to such mappings as well as their abundance.
\end{abstract}

\section{Introduction and Overview}
Codes over rings have been a major topic in coding theory ever since the discovery of the connection between linear codes over $\Z_4$ to some good non-linear binary codes in \cite{calderbank_1993}. They showed that there is a Gray map from codes over $\Z_4$ to codes over $\F_2^2$ for which the mentioned good non-linear codes were the images of linear  $\Z_4$ codes. Many classical constructions of codes over fields like cyclic codes and geometric codes have been generalized to rings  (\cite{bartley_2008,greferath_1997}). The idea of mapping codes from larger rings onto codes over smaller rings has been a growing topic where not only non-linear but also linear mappings are used.

\begin{example}
Consider the $\F_2$-algebra $A=\F_2[x]/(x^2)$ (see \cite{dougherty_1999_3}). Using the ordered $\F_2$ basis ${\cal B}=(1,x)$ of $A$, any $n$-length code over $A$ can be mapped to a $2n$-length code over $\F_2$ via the bijective Gray map
$$\Phi_{\cal B}: A^n\to(\F_2)^{2n};\; (a_1,a_2,\ldots,a_n)\to (\alpha_{1,1},\alpha_{1,2},\alpha_{2,1},\alpha_{2,2},\ldots,\alpha_{n,1},\alpha_{n,2})$$ where $a_i=\alpha_{i,1}+\alpha_{i,2}x$ is the representation of $a_i$ in the basis $(1,x)$. The image $\Phi_{\cal B}(C)$ of a linear code $C$ over $A$ is a code over $\F_2$. Specifically, consider $g=X^2 + xX + 1\in A[X]$. This polynomial is a divisor of $X^4-1\in A[X]$ and generates a cyclic code $C=(g)/(X^4-1)\subset A[X]/(X^4-1)$ of length $4$ over $A$. In the standard correspondence of $(g)/(X^4-1)$ with $A^4$, $g$ corresponds to the code word  $w_g=(1,x,1,0)$ which is mapped to  $\Phi_{\cal B}(w_g)=(1,0,0,1,1,0,0,0)$. Also, the code word $x\cdot (X^2 + xX + 1)$ is mapped to $(0,1,0,0,0,1,0,0)$. Applying this argument to the code word $X\cdot g$ we see that   $$\left(\begin{array}{cccccccc}1&0&0&1&1&0&0&0\\
0&1&0&0&0&1&0&0\\
0&0&1&0&0&1&1&0\\
0&0&0&1&0&0&0&1
\end{array}\right)$$
is a generator matrix of the binary image $\Phi_{\cal B}(C)$ of $C$.
\end{example}

Self-dual codes have become a central topic in coding theory due to their connections to other fields of mathematics such as block designs \cite{bonnecaze_2000}. In many instances, self-dual codes have been found by first finding a code over a ring and then mapping this code onto a code over a subring through a map that preserves duality. Often, these maps have been found through ad hoc methods. For instance, in  \cite{martinez_2015}, the local Frobenius non-chain rings of order 16 were found and a map that preserves duality was presented for each ring. In the literature, the mappings typically map to codes over $\F_2$, $\F_4$ and $\Z_4$ since codes over these rings have had the most use.

\begin{example} We consider again the cyclic code $C=(g)/(X^4-1)\subset A[X]/(X^4-1)$ of length $4$ over $A=\F_2[x]/(x^2)$ generated by $g(X)=X^2 + xX + 1\in A[X]$. Since $g$ equals its reciprocal polynomial $X^2\cdot g(1/X)$, the code $C$ is self-dual of length $4$ over $A$. However, the binary image $\Phi_{\cal B}(C)$ of $C$ obtained in previous example using the basis ${\cal B}=(1,x)$ is not a self-dual binary code. If we use the $\F_2$-basis ${\cal B}'=(1,x+1)$ to map to $\F_2^8$, then the image $\Phi_{\cal B'}(C)$ of $C$ is the type II binary code $[8,4,4]_2$ generated by
$$\left(\begin{array}{cccccccc}
1&0&1&1&1&0&0&0\\
0&1&1&1&0&1&0&0\\
0&0&1&0&1&1&1&0\\
0&0&0&1&1&1&0&1
 \end{array}\right).$$
Note that the code word $(1,x,1,0)$ corresponding to $g\in A[X]/(X^4-1)$ can be written  $(1\cdot 1+0\cdot(x+1),1\cdot 1+1\cdot(x+1),1\cdot 1+0\cdot(x+1),0\cdot 1+0\cdot(x+1))$ and maps to the first row  $(1, 0 ,1, 1 ,1, 0, 0 ,0)$, while  the second row is the image of $(1+x)\cdot g$. It is well known (\cite{dougherty_1999_3}) that in the $\F_2$-basis ${\cal B'}=(1,x+1)$ of $A$, the image of a self-dual code over $A$ under $\Phi_{\cal B'}$ is always a self-dual code over $\F_2$. Therefore, the mapping corresponding to the basis ${\cal B'}=(1,x+1)$ preserves duality.
\end{example}

A duality preserving Gray map does not always exist (see \cite{martinez_2015_3}). We aim to give criteria for bases which guarantee that the corresponding Gray map $\Phi$  preserves duality. Our setting is as described above: $A$ is a finite ring that is a free left module over a subring $R$ of $A$. We present two types of bases that have this duality preserving property. The first, {\it pseudo-self-dual bases}, are a generalization of trace orthogonal bases which have been defined for algebras over finite fields in \cite{lempel_1980} and \cite{betsumiya_2004}. The second, {\it symmetric bases}, are a generalization of the same for finite fields defined in \cite{mouaha_1992}. In the case of symmetric bases, it is required that $A$, not only be a left $R$-module, but a left $R$-algebra which in turn then requires that $R\subset Z(A)$.

Other than in \cite{alfaro_2015}, where a criterion is given to find a basis for $\F_q[x]/(x^t)$ that maps self-dual codes over $\F_q[x]/(x^t)$ to self-dual codes over $\F_q$, the authors are unaware of general methods for finding duality preserving maps. The methods herein are much more general than that in \cite{alfaro_2015}. It is shown here that the criteria for an $\F_q$-basis for $\F_q[x]/(x^t)$ to preserve duality given in \cite{alfaro_2015} is equivalent to the basis being symmetric (see Proposition \ref{prop_alf}).

The paper is organized as follows: Section \ref{sect_prelim} contains preliminaries and definitions needed throughout the rest of the paper. Section \ref{sect_duality} lays out the two methods for finding duality preserving bases. In section  \ref{sect_examples} we study the behavior of self-dual cyclic codes under duality preserving Gray maps and show that some Gray maps are better than others in terms of hamming distance for this family of codes.

\section{Preliminaries}
\label{sect_prelim}
For a ring $A$, $Aut(A)$ is the {\bf automorphism group} of $A$, a subset $C$ of $A^n$ is an {\bf $n$ length code over $A$} and if $C$ is a left $A$-submodule of $A^n$ then $C$ is an {\bf $A$-linear code} over the alphabet $A$. For a finite ring $A$ and a subgroup $H$ of $Aut(A)$, the {\bf fixed subring} of $H$ is $A^H=\{a\in A|h(a)=a\textrm{ for all }h\in H\}$ and the {\bf trace function with respect to $H$} is
\[
Tr_H:A\rightarrow A;\; a\mapsto\sum_{h\in H}h(a).
\]

\begin{lemma}
\label{lemma_trace}
Let $A$ be a ring and $H$ be a subgroup of $Aut(A)$. The trace function with respect to $H$ is both a left and right $A^H$-linear map and for $a\in A$, $Tr_H(a)\in A^H$.
\end{lemma}
\begin{proof}
Let $a\in A$, $b\in A^H$ and $g\in H$. Then
\[
Tr_H(ba)=\sum_{h\in H}h(ba)=\sum_{h\in H}h(b)h(a)=\sum_{h\in H}bh(a)=b\left(\sum_{h\in H}h(a)\right).
\]
So, $Tr_H$ is left $A^H$-linear. Similarly, $Tr_H$ is right $A^H$-linear. Also,
\[
g(Tr(a))=g\left(\sum_{h\in H}h(a)\right)= \sum_{h\in H}g(h(a))=\sum_{h\in H}h(a)=Tr(a)
\]
showing $Tr_H(a)\in A^H$.
\end{proof}

For a ring $A$, an anti-automorphism $\sigma$ on $A$ and a left $A$-module $M$, a {\bf $\sigma$-sesquilinear form} on $M$ is a map $\langle\cdot,\cdot \rangle:M\times M\to A$ such that if $x,y,z\in M$ and $a\in A$ then $\langle x+z,y\rangle=\langle x,y\rangle+\langle z,y\rangle$, $\langle x,y+z\rangle=\langle x,y\rangle+\langle x,z\rangle$, $\langle ax,y\rangle=a\langle x,y\rangle$ and $\langle x,ay\rangle=\langle x,y\rangle\sigma(a)$. In addition, if $\sigma(\langle x, y\rangle)=\langle y,x\rangle$ then the form is called a {\bf $\sigma$-hermitian form}. A $\sigma$-sesquilinear form with the property that $\id{x,y}=0\iff\id{y,x}=0$ is called {\bf reflexive}. Clearly a $\sigma$-hermitian form is reflexive.

\begin{proposition}
\label{prop_hform}
Let $A$ be a ring and $\sigma$ be an anti-automorphism on $A$. Define the map
\[
\id{\cdot,\cdot}:A^k\times A^k\to A;~~\id{x,y}=\sum_{i=1}^k x_i \sigma(y_i).
\]
Then $\id{\cdot,\cdot}$ is a $\sigma$-sesquilinear form. Furthermore, $\id{\cdot,\cdot}$ is a $\sigma$-hermitian form if and only if $\sigma$ is involutory, i.e. $\sigma^2=id$.
\end{proposition}
\begin{proof}
Clearly, $\id{\cdot,\cdot}$ is a $\sigma$-sesquilinear form. Assume $\id{\cdot,\cdot}$ is hermitian. For $a\in A$, $a=a\id{(1,0,\dots,0),(1,0,\dots,0)}=\id{(a,0,\dots,0),(1,0,\dots,0)}$. Since $\sigma$ is hermitian, $\sigma^2(a)=\sigma^2(\id{(a,0,\dots,0),(1,0,\dots,0)})=a$ showing $\sigma$ is involutory.

Now, assume $\sigma$ is involutory. For $x,y\in A^k$, $\sigma(\id{x,y})=\sigma(\sum_{i=1}^k x_i \sigma(y_i))=\sum_{i=1}^k \sigma(x_i \sigma(y_i))=\sum_{i=1}^k y_i \sigma(x_i)=\id{y,x}$ showing $\id{\cdot,\cdot}$ is $\sigma$-hermitian.
\end{proof}

When the form in Proposition \ref{prop_hform} is hermitian, it is known as the {\bf standard $\sigma$-hermitian form} on $A^k$ which we denote as $\id{x,y}_{A^k}$. If $\sigma$ is the identity map, $\id{x,y}_{A^k}$ is known as the {\bf standard bilinear form} on $A^k$. In the following an involution denotes an anti-automorphism of order $1$ or $2$. Since the identity map is an involution if and only if $A$ is commutative, we may only consider the standard bilinear form on $A^k$ when $A$ is commutative. If $R$ is a subring of $A$ such that $\sigma(R)=R$ then $\sigma|_{R}$ is an involution on $R$. Then, if entries in $A^k$ are restricted to $R^k$, the standard $\sigma$-hermitian form on $A^k$ is the standard $\sigma|_{R}$-hermitian form on $R^k$.

We will specialize to the standard $\sigma$-hermitian form on $A^k$ for which we define the orthogonal of a code. From Proposition \ref{prop_hform}, we know then that $\sigma$ is an involution. For a finite ring $A$, the standard $\sigma$-hermitian form on $A^k$ and an $A$-linear code $C\subset A^k$, the {\bf dual code} of $C$, denoted by $C^\perp$, is $C^\perp=\{v\,|\,\id{v,c}_{A^k}=0,\forall c\in C\}$. It is immediate that $C^\perp$ is a left $A$-module, therefore also an $A$-linear code. A code is {\bf self orthogonal} if $C\subset C^\perp$ and {\bf self dual} if $C=C^\perp$.

The dual we have defined is typically referred to as the left dual. There is an analog notion of a right dual. These of course when working over a commutative ring are identical. Furthermore, since we have defined the dual based on a hermitian form and a hermitian form is reflexive (see definition above), the left and right duals coincide over non-commutative rings as well, that is
\[
\{v\,|\,\id{v,c}_{A^k}=0,\forall c\in C\}=\{v\,|\,\id{c,v}_{A^k}=0,\forall c\in C\}.
\]

In \cite{wood_1999}, Wood showed, amongst other fundamental results on codes over Frobenius rings, that using the standard bilinear form to define the dual $C^\perp$ of an $n$-length linear code $C$ over a Frobenius ring $A$, $|C||C^\perp|=|A|^n$. This result was extended in \cite{szabo_2016} to bilinear forms and sesquilinear forms in general. The following is a specialization of that result to our setting.

\begin{lemma}
\label{lem_wood}
Let $A$ be a finite Frobenius ring and $C$ be a linear code over $A$. Then
$|C||C^\perp|=|A|^n$.
\end{lemma}
For this and other reasons presented in \cite{wood_1999}, most coding theorists restrict their study to codes over Frobenius rings. For the definition and details about Frobenius rings see \cite{wood_1999}. We will keep things general and not restrict the development of the theory to Frobenius until necessary. To illustrate one of the difficulties of working on codes over non-Frobenius rings we provide the following example.

\begin{example}
Let $A=\F_2[u,v]/(u^2,v^2,uv)$ which is of order 8. The Jacobson radical of $A$ is $J(A)=(u,v)$ and the socle of $A$ is $S(A)=(u,v)$. A finite Frobenius ring $R$ can be characterized by having the property that $R/J(R)\cong S(R)$ as left $R$-modules. We see that as $A$-modules, $A/J(A)=\{J(A),1+J(A)\}\cong\F_2$ but $Soc(R)=(u)+(u+v)+(v)\cong\F_2^3$. So, $A$ is non-Frobenius. Consider the ideal $C=(u,v)\lhd A$ as a 1-length code over $A$. We see that $C=C^\perp$. So, $|C||C^\perp|=16>8=|A|$.
\end{example}

\section{Duality Preserving Bases}
\label{sect_duality}
Throughout this section let $n\in \mathbb{N}\setminus\{0\}$, let $A$ be a finite unitary ring which is a free left module over a unitary subring $R$ of $A$, $\mathcal{B}=(v_1,\dots,v_r)$ be an ordered left $R$-basis for $A$ and $\sigma$ be an involution on $A$ such that $\sigma(R)=R$.  Define the  maps
\[
\rho:A\rightarrow R^r;~~\alpha_1v_1+\dots+\alpha_rv_r\mapsto(\alpha_1,\dots,\alpha_r)
\]
and
\[
\varPhi:A^n\rightarrow R^{rn};~~(a_1,\dots,a_n)\mapsto(\rho(a_1),\dots,\rho(a_n)).
\]

Typically, the map $\varPhi$ is called a Gray map. With it, we define the \textbf{Gray weight on $A$ with respect to $\mathcal{B}$} as follows: For $z\in A^n$, $W_\mathcal{B}(z)=w_H(\varPhi(z))$ where $w_H$ is the Hamming weight on $R^{rn}$. Given a code $C$ over $A$, for every $R$-basis of $A$ there is an image viewed through this basis and a weight function with respect to this basis for which $\varPhi$ becomes an isometry. The next example illustrates this.

\begin{example}
Assume $A=\F_2\times\F_2$. Any two non-zero elements of $A$ form an $\F_2$-basis. Let $\mathcal{C}=\{(1,0),(0,1)\}$, $\mathcal{D}=\{(1,0),(1,1)\}$ and $\mathcal{E}=\{(1,1),(0,1)\}$. Then for the codeword of length 2 over $A$, $c=((1,0),(1,0),(0,1))$, $\varPhi_\mathcal{C}(c)=(1,0,1,0,0,1)$, $\varPhi_\mathcal{D}(c))=(1,0,1,0,1,1)$ and $\varPhi_\mathcal{E}(c))=(1,1,1,1,1,0)$. So, $W_\mathcal{C}(c)=3$, $W_\mathcal{D}(c)=4$ and $W_\mathcal{E}(c)=5$.
\end{example}

Before moving on, we show the connection between the inner product on $A^n$ and the inner product on $R^{rn}$. To that end we introduce the following, let
\[
M=\mathcal{B}^{^T}\sigma(\mathcal{B})=\left(
                                        \begin{array}{ccc}
                                          v_1\sigma(v_1) & \dots & v_1\sigma(v_r) \\
                                          \vdots & \vdots & \vdots \\
                                          v_r\sigma(v_1) & \dots & v_r\sigma(v_r) \\
                                        \end{array}
                                      \right),
\]
and let $\mathcal{M}$ be the block diagonal $nr\times nr$ matrix with $M$ on the diagonal.

\begin{lemma}
\label{lemma_form}
Let $x,y\in A^n$. Then $
\id{x,y}_{A^n}=\varPhi(x)\mathcal{M}\sigma(\varPhi(y))^T$.
\end{lemma}
\begin{proof}
Let $x_i$ and $y_i$ be the $i$-th component of $x$ and $y$ respectively. Note that for $a\in A$, $a=\rho(a)\mathcal{B}^{^T}$. Then
\begin{eqnarray*}
\id{x,y}_{A^n}&=&\sum_{i=1}^n x_i \sigma(y_i)=\sum_{i=1}^n\rho(x_i)\mathcal{B}^{^T}\sigma(\rho(y_i)\mathcal{B}^{^T})=\sum_{i=1}^n \rho(x_i)\mathcal{B}^{^T}\sigma(\mathcal{B})\sigma(\rho(y_i))^{^T}\\
&=&\sum_{i=1}^n \rho(x_i)M\sigma(\rho(y_i))^{^T}=\varPhi(x)\mathcal{M}\sigma(\varPhi(y))^{^T}
\end{eqnarray*}
\end{proof}
In this section we investigate the sufficient conditions on $\mathcal{B}$ so that $\id{x,y}_{A^n}=0$ implies $\id{\varPhi(x),\varPhi(y)}_{R^{nr}}=0$ giving a so called duality preserving basis. We will give two separate conditions on $\mathcal{B}$ that guarantee it will preserve duality.

\subsection{Pseudo-Self-dual Bases}
\label{sect_psd}
In this subsection we consider pseudo-self-dual bases and show that such bases preserve duality. In \cite{lempel_1980}, trace orthogonal bases for finite field extensions were defined. We extend this definition to include ring extensions in our setting as we consider the extension $A\supset R$.

\begin{definition}
For a subgroup $H$ of $Aut(A)$ we define the following. $\mathcal{B}$ is a {\bf $\sigma$-trace orthogonal basis} with respect to $H$ if for $1\leq i,j\leq r$, $Tr_H(v_i\sigma(v_j))=0$ if and only if $i\neq j$. In addition, if there exists $\gamma\in A$ that is not a zero divisor, commutes with elements of $R$ and $Tr_H(v_i\sigma(v_i))=\gamma$ for $1\leq i\leq r$ then $\mathcal{B}$ is called a {\bf $\sigma$-pseudo-self-dual basis} with respect to $H$. Furthermore, if $\gamma=1$, $\mathcal{B}$ is called a {\bf $\sigma$-self-dual basis} with respect to $H$.
\end{definition}

\begin{table}[!htp]
\renewcommand{\arraystretch}{1.5}
\newcommand{\wb}[1]{\colorbox{white}{#1}}
\setlength{\tabcolsep}{3pt}
\caption{\small Number of $\sigma$-pseudo-self-dual bases w.r.t. $H$ of $\F_3(\xi)[x]/(x^2+1)$ over $R$}
{\small\centering
\label{table_1}
\begin{subtable}{1\textwidth}
\begin{tabular}{|c||c|c|c|c|c|c|c|c|c|}
\hline
\backslashbox{$\sigma$}{$H$}&\colorbox{Gray}{$\id{\psi}$}&\wb{$\id{\theta^2}$}&\wb{$\id{\theta^2\psi}$}&\wb{$\id{\theta\psi}$}&\wb{$\id{\theta^3\psi}$}&\wb{$\id{\theta}$}&\wb{$\id{\psi,\theta^2}$}&\wb{$\id{\theta\psi,\theta^2}$}&\wb{$Aut(A)$}\\\hline\cline{2-10}
\colorbox{Gray}{$id$}&64&0&32&0&0&112&112&0&112\\\hline
\colorbox{Gray}{$\psi$}&64&16&64&0&0&112&112&16&112\\\hline
\colorbox{Gray}{$\theta^2$}&96&64&64&0&0&96&96&64&96\\\hline
\colorbox{Gray}{$\theta^2\psi$}&96&32&0&0&0&96&96&32&96\\\hline
$\theta\psi$&32&0&32&0&0&352&352&0&352\\\hline
$\theta^3\psi$&32&0&32&0&0&352&352&0&352\\\hline
\end{tabular}
\caption{$R=\F_3(\xi)$\bigskip}
\end{subtable}
\begin{subtable}{1\textwidth}
\begin{tabular}{|c||c|c|c|c|c|c|c|c|c|}
\hline
\backslashbox{$\sigma$}{$H$}&\wb{$\id{\psi}$}&\wb{$\id{\theta^2}$}&\colorbox{Gray}{$\id{\theta^2\psi}$}&\wb{$\id{\theta\psi}$}&\wb{$\id{\theta^3\psi}$}&\wb{$\id{\theta}$}&\wb{$\id{\psi,\theta^2}$}&\wb{$\id{\theta\psi,\theta^2}$}&\wb{$Aut(A)$}\\
\hline\cline{2-10}
\colorbox{Gray}{$id$}&32&0&64&0&0&112&112&0&112\\\hline
\colorbox{Gray}{$\psi$}&0&32&96&0&0&96&96&32&96\\\hline
\colorbox{Gray}{$\theta^2$}&64&64&96&0&0&96&96&64&96\\\hline
\colorbox{Gray}{$\theta^2\psi$}&64&16&64&0&0&112&112&16&112\\\hline
$\theta\psi$&32&0&32&0&0&352&352&0&352\\\hline
$\theta^3\psi$&32&0&32&0&0&352&352&0&352\\
\hline
\end{tabular}
\caption {$R=\F_3(x+1)$\bigskip}
\end{subtable}
\begin{subtable}{1\textwidth}
\begin{tabular}{|c||c|c|c|c|c|c|c|c|c|}
\hline
\backslashbox{$\sigma$}{$H$}&\colorbox{Gray}{$\id{\psi}$}&\colorbox{Gray}{$\id{\theta^2}$}&\colorbox{Gray}{$\id{\theta^2\psi}$}&\colorbox{Gray}{$\id{\theta\psi}$}&\colorbox{Gray}{$\id{\theta^3\psi}$}&\colorbox{Gray}{$\id{\theta}$}&\colorbox{Gray}{$\id{\psi,\theta^2}$}&\colorbox{Gray}{$\id{\theta\psi,\theta^2}$}&\colorbox{Gray}{$Aut(A)$}\\
\hline\cline{2-10}
\colorbox{Gray}{$id$}&0&0&0&0&0&96&96&0&96\\\hline
\colorbox{Gray}{$\psi$}&0&0&0&0&0&96&96&0&96\\\hline
\colorbox{Gray}{$\theta^2$}&0&0&0&0&0&96&96&0&96\\\hline
\colorbox{Gray}{$\theta^2\psi$}&0&0&0&0&0&96&96&0&96\\\hline
\colorbox{Gray}{$\theta\psi$}&0&0&0&0&0&0&0&0&0\\\hline
\colorbox{Gray}{$\theta^3\psi$}&0&0&0&0&0&0&0&0&0\\
\hline
\end{tabular}
\caption {$R=\F_3$\bigskip}
\end{subtable}}
\begin{tablenotes}
\item A gray highlighted involution $\sigma$ indicates that $\sigma(R)=R$ and a gray highlighted subgroup $H$ indicates that $H$ fixes $R$ element-wise.
\end{tablenotes}
\end{table}

\begin{example}
\label{ex_main}
Assume $A=\F_3(\xi)[x]/(x^2+1)$ where $\xi^2+2\xi+2=0$. Note, $\F_3(\xi)\cong\F_9$. Since $x^2+1=(x + \xi^2)(x + \xi^6)\in\F_3(\xi)[x]$, we see that
\[
A\cong\F_3(\xi)[x]/(x+\xi^2)\oplus\F_3(\xi)[x]/(x+\xi^6)\cong\F_9\oplus\F_9.
\]
The automorphism group of $A$ is the dihedral group $D_4$ of order $8$ generated by the automorphism $\theta:A\to A;\xi\mapsto x+2,x\mapsto \xi^6$ of order 4 and $\psi:A\to A;\xi\mapsto \xi,x\mapsto 2x$ of order 2. So, there are 5 automorphism of order 2: $\psi$, $\theta\psi$, $\theta^2\psi$, $\theta^3\psi$ and $\theta^2$.
Outside of the subgroups of order 2 generated by the elements of order 2, there are 3 proper subgroups, $Aut(A)$, $\id{\theta}$, $\id{\psi,\theta^2}$ and $\id{\theta\psi,\theta^2}$ each of which is of order 4. Using Magma, we have found the number of $\sigma$-pseudo-self-dual bases of $A$ over $R$ with respect to $H$ for each involution of $A$, each subgroup $H$ of $Aut(A)$ and each subring $R$ of $A$. The results are presented in Table \ref{table_1}. In the table, a gray highlighted involution $\sigma$ indicates that $\sigma(R)=R$ and a gray highlighted subgroup $H$ indicates that $H$ fixes $R$ element-wise (which, according to theorem \ref{theo_tob},  will insure a duality preserving basis for the corresponding Gray map $\Phi$). Note that some bases will be found with multiple subgroups.
\end{example}

\begin{example}
\label{GRf} Assume $A=GR(4,2)$, the Galois ring consisting of all elements of the form $\beta_0+\beta_1\xi $ where $\beta_i\in \Z_4$ and $\xi^2+\xi+1=0$. The automorphism group of $A$ is of order $2$ generated by $\theta:A\to A$ defined by $\theta:\xi \mapsto 3\xi+3$. Assume $\sigma=\theta$, $\mathcal{B}=(\xi,\xi+3)$ and $H=Aut(A)=\{id,\theta\}$. Now, $Tr_H(\xi\sigma(\xi))=3$, $Tr_H(\xi\sigma(\xi+3))=0$, $Tr_H((\xi+3)\sigma(\xi))=0$ and $Tr_H((\xi+3)\sigma(\xi+3))=3$ showing that $\mathcal{B}$ is a $\sigma$-pseudo-self-dual basis where $\gamma=3$. With a similar verification we see that if $\sigma=id$, $\mathcal{B}$ is a $\sigma$-pseudo-self-dual basis where $\gamma=3$ as well. In \cite{boucher_2008}, $\{\xi,\xi+3\}$ was used to map Euclidean self-dual codes over $A$ to Euclidean self-dual codes over $\Z_4$.

Using Magma we found that with $\sigma=id$ or $\sigma=\theta$, $A$ has no $\sigma$-self-dual basis over its prime ring $\Z_4$, but has $8$ $\sigma$-pseudo-self-dual bases over $\Z_4$:
\[
\{3\xi+2,\xi+1\},
\{\xi+1,\xi+2\},
\{3\xi+2,3\xi+3\},
\{3\xi+1,\xi\}
\]
\[
\{3\xi+3,\xi+2\},
\{\xi,\xi+3\},
\{3\xi,\xi+3\},
\{3\xi,3\xi+1\}.
\]
\end{example}

The previous examples show that $\sigma$-pseudo-self-dual bases have been used in the literature and also that even for small algebras they may be abundant. We soon give the main result of this subsection which shows that under a $\sigma$-pseudo-self-dual bases, the image of the dual of a code is the dual of the image i.e. preserves duality given some additional conditions. The first condition is that $\sigma(R)=R$ which we have already imposed throughout this section. This is to guarantee that the $\sigma$ hermitian form on $A$ restricts to a hermitian form on $R$. The second condition is that $R\subset A^H$. The next example shows that with out this, a $\sigma$-pseudo-self-dual basis does not necessarily preserve duality.

\begin{example}
\label{ex_main2}
We continue in the setting of Example \ref{ex_main}. Assume $R=\F_3(\xi)$ and $\sigma=\psi$ (notice $\psi(\F_3(\xi))=\F_3(\xi)$). Assume $\mathcal{B}=\{1,\xi^2x+\xi^2\}$. The $R$-basis of $A$, $\mathcal{B}$, is a $\sigma$-pseudo-self-dual basis with respect to $Aut(A)$. Notice $R\not\subset A^{Aut(A)}$. Let $C$ be the $[2,1]$ linear code over $A$ generated by
\[
\left(\begin{array}{cccc}
1&\xi^2x\\
\end{array} \right)
\]
which is a $\sigma$-hermitian self-dual code over $A$. The code $\varPhi(C)$, whose generator matrix is,
\[
\left(\begin{array}{cccc}
1&0&\xi^6&1\\
0&1&2&\xi^2
\end{array}\right)
\]
is not a euclidian self-dual code over $R$.
\end{example}

The next theorem is the first of our two main theorems, each of which provides conditions for mapping the dual of a code over $A$ to the dual of its image over $R$.

\begin{theorem}
\label{theo_tob}
Let $H$ be a subgroup of the automorphism group of $A$ such that $R\subset A^H$ and let $C$ be an $n$ length linear code over $A$. Assume $\mathcal{B}$ is $\sigma$-pseudo-self-dual basis w.r.t $H$. Then
\[
\varPhi(C^\perp)\subset\varPhi(C)^\perp.
\]
Furthermore, if $A$ and $R$ are Frobenius rings then
\[
\varPhi(C^\perp)=\varPhi(C)^\perp.
\]
\end{theorem}
\begin{proof}
Since $\mathcal{B}$ is $\sigma$-pseudo-self-dual basis w.r.t $H$ there exists a $\gamma\in A$ that commutes with $R$, is not a zero divisor and for $1\leq i,j\leq r$ with $i\neq j$,  $Tr_H(v_i\sigma(v_i))=\gamma$ and $Tr_H(v_i\sigma(v_j))=0$. Let $x=(x_1,\ldots,x_n)\in A^n$ and $y=(y_1,\ldots,y_n)\in A^n$ where $x_i=\sum_{j=1}^r\alpha_{ij}v_j\in A$ and $y_i=\sum_{k=1}^r\beta_{ik}v_k\in A$. Assume $\id{ x,y}_{A^n}=0$. Since $R\subset A^H$, by Lemma \ref{lemma_trace} we have that

\begin{eqnarray*}
0&=&Tr_H(0)=Tr_H(\id{ x,y}_{A^n})=Tr_H\left(\sum_{i=1}^n\left(\sum_{j=1}^r\alpha_{ij}v_j\right)\sigma\left(\sum_{k=1}^r\beta_{ik}v_k\right)\right)\\
&=&Tr_H\left(\sum_{i=1}^n\left(\sum_{j=1}^r\alpha_{ij}v_j\right)\left(\sum_{k=1}^r\sigma\left(v_k\right)\sigma\left(\beta_{ik}\right)\right)\right)\\
&=&\sum_{i=1}^n\sum_{j=1}^r\sum_{k=1}^r\alpha_{ij}Tr_H\left(v_j\sigma\left(v_k\right)\right)\sigma\left(\beta_{ik}\right)=\sum_{i=1}^n\sum_{j=1}^r\alpha_{ij}Tr_H\left(v_j\sigma\left(v_j\right)\right)\sigma\left(\beta_{ij}\right)\\
&=&\sum_{i=1}^n\sum_{j=1}^r\alpha_{ij}\gamma\sigma\left(\beta_{ij}\right)=\gamma\sum_{i=1}^n\sum_{j=1}^r\alpha_{ij}\sigma\left(\beta_{ij}\right)\\
&=&\gamma\id{\varPhi(x),\varPhi(y)}_{R^{nr}}.
\end{eqnarray*}
Since $\gamma$ is not a zero divisor, $\id{\varPhi(x),\varPhi(y)}_{R^{nr}}=0$. With this it is easy to see that $\varPhi(C^\perp)\subset\varPhi(C)^\perp$.

Now, $|C|=|\varPhi(C)|$. Since $A$ and $R$ are Frobenius rings, from Lemma \ref{lem_wood}, we have that
\[
|C^\perp|={|A|^n\over |C|}={|R|^{rn}\over |\varPhi(C)|}=|\varPhi(C)^\perp|.
\]
Hence, $\varPhi(C^\perp)=\varPhi(C)^\perp$.
\end{proof}

For an alternate view of the previous result, we return to the notation introduced at the beginning of this section. From Lemma \ref{lemma_form} we have that
\[
\id{x,y}_{A^n}=\varPhi(x)\mathcal{M}\sigma(\varPhi(y))^{^T}.
\]
In the setting of Theorem \ref{theo_tob} we have then
\begin{eqnarray*}
0&=&Tr_H(0)=Tr_H(\id{x,y}_{A^n})=Tr_H(\varPhi(x)\mathcal{M}\sigma(\varPhi(y)^{^T}))\\
&=&\varPhi(x)Tr_H(\mathcal{M})\sigma(\varPhi(y)^{^T})=\varPhi(x)\gamma I_{nr}\sigma(\varPhi(y)^{^T})=\gamma\id{\varPhi(x),\varPhi(y)}_{R^{nr}}.
\end{eqnarray*}
Again, since $\gamma$ is not a zero divisor, $\id{\varPhi(x),\varPhi(y)}_{R^{nr}}=0$. The point here is this. In general, if $Tr_H(\mathcal{M})=\gamma I_{nr}$ which boils down to $Tr_H(M)=\gamma I_{r}$, then $\mathcal{B}$ is a $\sigma$-pseudo-self-dual basis for $A$ over $R$. If in addition, $R\subset A^H$, $\mathcal{B}$ will preserve duality.

\begin{example}
We again continue in the setting of Examples \ref{ex_main} and \ref{ex_main2}. Remember $A=\F_3(\xi)[x]/(x^2+1)$ where $\xi^2+2\xi+2=0$, $R=\F_3(\xi)$ and $\sigma:A\to~A;\xi\mapsto \xi,x\mapsto 2x$. Here we assume $\mathcal{B}=\{1,x\}$. As in Example \ref{ex_main2}, the basis here is a $\sigma$-pseudo-self-dual basis but it is with respect to $\id{\sigma}\subset Aut(A)$. What is different is that $R$ is in the fixed ring of the subgroup i.e. $R\subset A^{\id{\sigma}}$. Remember, the code $C$ generated by
\[
\left(\begin{array}{cccc}
1&\xi^2x\\
\end{array} \right)
\]
is a $\sigma$-hermitian self-dual code. With the redefinition of $\mathcal{B}$, the code $\varPhi(C)$, will have generator matrix
\[
\left(\begin{array}{cccc}
1&0&0&\xi^2\\
0&1&\xi^6&0
\end{array}\right).
\]
But, unlike the image of this code in Example \ref{ex_main2}, the image here is a euclidian self-dual code over $A$.
\end{example}

It is not always the case that $A$ has a $\sigma$-pseudo-self-dual basis over $R$. For instance, from the full list of indecomposable commutative rings of order 16 given in \cite{martinez_2015}, none of the rings $\F_4[x]/(x^2)$,  ${\F_2[x]/(x^2+y^2,xy)}$, $\F_2[x]/(x^2,y^2)$, ${\Z_4[x]/(x^2+2x)}$, ${\Z_4[x]/(x^2+2)}$, $\Z_4[x]/(x^2+2x+2)$ and $\Z_4[x]/(x^2)$ have a $\sigma$-pseudo-self-dual basis for any proper subring. This is not to say that no duality preserving map exists over these rings as will be seen in the next subsection where we look at an alternative property which guarantees duality preservation. In the case of ${\Z_4[x]/(x^2)}$ it is true that no such duality preserving map exists as this was shown in \cite{martinez_2015_3}. Similarly, we can find examples where $A$ is noncommutative and does not have a $\sigma$-pseudo-self-dual basis over $R$. For instance, $\frac{\F_4[x;\theta]}{\id{x^2}}$ where $\theta$ is the Frobenius map on $\F_4$ extended to $\frac{\F_4[x;\theta]}{\id{x^2}}$.

The next two results have to do with scaling $\sigma$-pseudo-self-dual bases in order to obtain other $\sigma$-pseudo-self-dual bases. In Proposition \ref{prop_sc1}, given a $\sigma$-pseudo-self-dual basis, another $\sigma$-pseudo-self-dual basis was found. In Proposition \ref{prop_TOB}, given an involution $\varphi$ on $A$ and a $\sigma$-pseudo-self-dual basis, a $\varphi\sigma\varphi$-pseudo-self-dual basis is found i.e. the new basis is pseudo-self-dual for the $\varphi\sigma\varphi$ sesquilinear form.

\begin{proposition}
\label{prop_sc1}
Let $H$ be a subgroup of the automorphism group of $A$ and $a\in A^H$ such that $\sigma(a)\in A^H$, that is not a zero divisor and commutes with the elements of $R$. Assume $\{v_1,\ldots,v_n\}$ is $\sigma$-pseudo-self-dual basis with respect to $H$ where for some $\gamma\in A^H$, $Tr_H(v_i\sigma(v_i))=\gamma$. Then $\{av_1,\ldots,av_n\}$ is a $\sigma$-pseudo-self-dual basis for $\gamma'=a\gamma\sigma(a)$ with respect to $H$.
\end{proposition}
\begin{proof}
Since $a,\sigma(a)\in A^H$, by Lemma \ref{lemma_trace} we have that
\begin{eqnarray*}
Tr_H(v_i\sigma(v_i))&=&\gamma\\
aTr_H(v_i\sigma(v_i))\sigma(a)&=&a\gamma\sigma(a)\\
Tr_H(av_i\sigma(v_i)\sigma(a))&=&a\gamma\sigma(a)\\
Tr_H(av_i\sigma(av_i))&=&a\gamma\sigma(a).\\
\end{eqnarray*}
Since $a$ is regular, so is $\sigma(a)$ and $a\gamma\sigma(a)$. The result follows.
\end{proof}

\begin{proposition}
\label{prop_TOB}
Let $H$ be a subgroup of $Aut(A)$ and $\varphi$ be an involution on $A$. Denote by $\theta$, the automorphism $\varphi\sigma$ of $A$, by $\widehat{\sigma}$, the involution $\varphi\sigma\varphi$ of $A$, by $\widehat{H}$ the subgroup $\varphi H\varphi$ of ${\rm Aut}(A)$ and by $\widehat{{\mathcal{B}}}$, the $\theta(R)$-basis of $A$, $\{\theta(v_1),\ldots,\theta(v_n)\}$.  Then
\begin{enumerate}
\item For $\widehat{v}=\theta(v)$ and $\widehat{w}=\theta(w)$ in $\widehat{\mathcal{B}}$, $Tr_{\widehat{H}}(\widehat{v}\sigma(\widehat{w}))=\varphi\left(Tr_{H}(w\sigma(v))\right)$.\label{lemma_inv.four}
\item If $\mathcal{B}$ is $\sigma$-pseudo-self-dual basis w.r.t $H$ over $R$ of $A$ then $\widehat{{\cal B}}$ is a $\widehat{\sigma}$-pseudo-self-dual basis w.r.t $\widehat{H}$ over $\theta(R)$ of $A$.\label{lemma_inv.five}
\end{enumerate}
\end{proposition}
\begin{proof} The first statement follows by  a simple computation:
\begin{eqnarray*}
Tr_{\widehat{H}}(\widehat{v}\sigma(\widehat{w}))&=&\sum_{h\in\widehat{H}}h(\widehat{v}\widehat{\sigma}(\widehat{w})=\sum_{h\in H}\varphi h\varphi(\varphi\sigma(v)\varphi\sigma\varphi(\varphi\sigma(w))\\
&=&\varphi\sum_{h\in H} h(w\sigma(v))=\varphi\left(Tr_{H}(w\sigma(v))\right).
\end{eqnarray*}
For the second statement one verifies that  $\widehat{{\cal B}}$ is a left $\theta(R)$-basis of $A$.
Suppose that   $\mathcal{B}$ is $\sigma$-pseudo-self-dual basis w.r.t $H$ for $\gamma=\left(Tr_{H}(v_i\sigma(v_i))\right)$, where $\gamma$ is not a zero divisor that commutes with the elements of $R$.
From the first statement we obtain for  $1\leq i,j\leq r$ that $Tr_{\widehat{H}}(\theta(v_i)\sigma(\theta(v_j)))=0$ if $i\neq j$ and $Tr_{\widehat{H}}(\theta(v_i)\sigma(\theta(v_i)))=\varphi(\gamma)$. Since $\gamma$ is not a zero divisor, $\varphi(\gamma)$ is not a zero divisor and it only remains to show that  $\varphi(\gamma)$ commutes with elements of $\varphi(R)$. For $r\in R$, since $\sigma(R)=R$ and $\gamma$ commutes with elements of $R$,
$$\varphi(\gamma)\theta(r)=\varphi(\gamma)\varphi(\sigma(r))=\varphi(\sigma(r)\gamma)=\varphi(\gamma\sigma(r))=\varphi(\sigma(r))\varphi(\gamma)=\theta(r)\varphi(\gamma)$$
 and the result follows.\end{proof}

\begin{example}
In Table \ref{table_1}, the table associated with Example \ref{ex_main},  we see that starting with the table for $R=\F_3(\xi)$, swapping rows 2 and 4 and then columns 1 and 3 we end up with the table  for $R=\F_3(x+1)$. This can be explained using Proposition \ref{prop_TOB}. Notice that the number of $\psi$-pseudo-self-dual bases with respect to $\id{\theta^2\psi}$  of $A$ over $\F_3(\xi)$ is 64. Now, $(\theta\psi)\psi(\theta\psi)=\theta^2\psi$, $(\theta\psi)\id{\theta^2\psi}(\theta\psi)=\id{\psi}$ and $(\theta\psi)\psi=\theta$. Since $\theta(\F_3(\xi))=\F_3(x+1)$, Proposition \ref{prop_TOB} shows that there are 64 $\theta^2\psi$-pseudo-self-dual bases with respect to $\id{\psi}$ of $A$ over $\F_3(x+1)$.
\end{example}

\begin{example}
In Example \ref{m22b}, a full analysis is given of $\sigma$-pseudo-self-dual bases on $M_2(\F_2)$. Here we refer to that example. Since $\psi$ conjugated with any involution is $\psi$, applying Proposition \ref{prop_TOB} to a $\psi$-pseudo-self-dual basis will produce a $\psi$-pseudo-self-dual basis. Similarly, $\tau$ conjugated with $\tau$ or $\psi$ is $\tau$, applying Proposition \ref{prop_TOB} in these cases will only produce a $\tau$-pseudo-self-dual basis from $\tau$-pseudo-self-dual basis. More interestingly though, $(\tau\theta)\tau(\tau\theta)=\tau\theta^2$. Now if we apply Proposition \ref{prop_TOB} with the involution $\tau\theta$ to a $\tau$-pseudo-self-dual basis, we obtain a $\tau\theta^2$-pseudo-self-dual basis. Similarly, $(\tau\theta^2)\tau(\tau\theta^2)=\tau\theta$. If we apply Proposition \ref{prop_TOB} with the involution $\tau\theta^2$ to a $\tau$-pseudo-self-dual basis, we obtain a $\tau\theta$-pseudo-self-dual basis. This explains why there are the same number of $\tau$, $\tau\theta$ and $\tau\theta^2$ pseudo-self-dual bases.
\end{example}

\subsection{Symmetric Bases}
In this section assume additionally that $A$ is a free $R$-algebra not simply a left $R$-module. Although we do not assume $A$ is commutative here, if we assume the $R$ basis  $\mathcal{B}$ to be symmetric (defined below), then $A$ must be commutative. See Lemma \ref{lemma_mult}. So, this section in reality is strictly about commutative alphabets. In the commutative case, $\sigma$ is simply an order two automorphism or the identity map. This allows us to consider the Euclidean inner product on commutative rings in our setting which is not possible in the non-commutative case as the identity map is not an involution in that case.

In \cite{mouaha_1992}, symmetric bases were defined for field extensions. As we did with trace orthogonal bases in the last section, we extend the definition of symmetric bases to include the ring extensions we are considering.

\begin{definition}
\label{def_emb}
For $a\in A$, let $M_a$ denote the matrix w.r.t. $\mathcal{B}$ representing the linear transformation of right multiplication by $a$, i.e.
\[
M_a=\left[
      \begin{array}{c}
        \rho(v_1a) \\
        \vdots \\
        \rho(v_ra) \\
      \end{array}
    \right].
\]
We say $\mathcal{B}$ is a \textit{symmetric basis} if for any $v\in\mathcal{B}$, $M_{v}=M_{v}^{^T}$.
\end{definition}

\begin{example}
\label{ex_eff}
Assume $A=\F_2(\xi)[x]/(x^2)$ where $\xi^2+\xi+1=0$ ($\F_2(\xi)\cong\F_4$) and $\mathcal{B}=(1,\xi x+1)$ which is an $\F_2(\xi)$-basis for $A$. See \cite{ling_2001_2} for specific work on codes over $\F_4[x]/(x^2)$. Using the embedding described in Definition \ref{def_emb}, we have $M_1=\left(\begin{array}{cc}1&0\\0&1\end{array}\right)$ and $M_{x+1}=\left(\begin{array}{cc}0&1\\1&0\end{array}\right)$ showing that $\mathcal{B}$ is symmetric.
\end{example}

In the above definition $M_a$ acts on row vectors from $R^r$ on the right. It is well know that sending $a\in A$ to $M_a$ is an embedding of $A$ in $M_r(R)$. Before getting to the main result of the section, we have a lemma with a few simple results needed throughout. The most important of these is that for $A$ to have a symmetric basis, it must be commutative.

\begin{lemma} \label{lemma_mult}
\begin{enumerate}
\item For $a,b\in A$, $\rho(ab)=\rho(a)M_b$ and $ab=\rho^{-1}(\rho(a)M_b)$.
\item If $\mathcal{B}$ is symmetric, then for $a\in A$, $M_a=M_a^{^T}$ and $A$ is commutative.
\end{enumerate}
\end{lemma}
\begin{proof}
Let $a,b\in A$. Since $\rho(a)$ is just the representation of $a$ in $\mathcal{B}$, $\rho(ab)=\rho(a)M_b$ and $ab=\rho^{-1}(\rho(a)M_b)$. If $\mathcal{B}$ is symmetric, since $\{M_v|v\in\mathcal{B}\}$ is a basis for $\{M_a|a\in A\}$ we have that $M_a=M_a^{^T}$ and so
\[
M_{ab}=M_aM_b=(M_aM_b)^{^T}=M_b^{^T}M_a^{^T}=M_bM_a=M_{ba}.
\]
Finally, $ab=ba$.
\end{proof}

As was the case with pseudo-self-dual bases, symmetric bases also preserve duality given an additional condition which is the main result of this section. In addition to a basis being symmetric, for the basis to guarantee duality preservation, $\sigma$ must fix the basis element-wise.

\begin{example}
\label{ex_symmain}
Continuing with Example \ref{ex_eff} where $A=F_2(\xi)[x]/(x^2)$ and $\mathcal{B}=(1,\xi x+1)$, assume $\sigma$ is the Frobenius map on $\F_2(\xi)$ extended linearly to $A$. Remember, $\mathcal{B}$ is symmetric and observe that $\sigma(1)=1$ and $\sigma(\xi x+1)=\xi^2 x+1$ meaning $\sigma$ does not fix  $\mathcal{B}$ element-wise.
Let $C$ be the $[2,1]$ linear code over $A$ generated by
\[
\left(\begin{array}{cccc}
1&1+x\\
\end{array} \right)
\]
which is a $\sigma$-hermitian self-dual code over $A$. The code $\varPhi(C)$, whose generator matrix is,
\[
\left(\begin{array}{cccc}
1&0&\xi&\xi^2\\
0&1&\xi^2&\xi
\end{array}\right)
\]
is not a $\sigma$-hermitian self-dual code over $\F_2(\xi)$.
\end{example}

This next theorem, the second of our two main results, shows that if a basis is symmetric and is fixed by $\sigma$ element-wise, the basis preserves duality.

\begin{theorem}
\label{theo_sym}
Let $C$ be an $n$ length linear code over $A$. Assume $\mathcal{B}$ is symmetric and $\sigma$ fixes the elements of $\mathcal{B}$. Then
\[
\varPhi(C^\perp)\subset\varPhi(C)^\perp.
\]
Furthermore, if $A$ and $R$ are Frobenius rings then
\[
\varPhi(C^\perp)=\varPhi(C)^\perp.
\]
\end{theorem}
\begin{proof}
Let $a=\alpha_1v_1+\dots+\alpha_rv_r,b=\beta_1v_1+\dots+\beta_rv_r\in A$ where $\alpha_i,\beta_i\in R$. Since $\sigma$ fixes the elements of $\mathcal{B}$, $(\sigma(\beta_1),\dots,\sigma(\beta_r))=\rho(\sigma(b))=\rho(1)M_{\sigma(b)}$. So,
\[
\id{\rho(a),\rho(b)}_{R^r}=\sum_{i=1}^r \alpha_i\sigma(\beta_i)=\rho(1)M_a\left(\rho(1)M_{\sigma(b)}\right)^{^T}.\\
\]
Since $\langle a,b\rangle=0$, $0=\sum_{i=1}^n\alpha_i\sigma(\beta_i)$ which implies $0=\sum_{i=1}^nM_{\alpha_i}M_{\sigma(\beta_i)}$. Since $\mathcal{B}$ is symmetric, by Lemma \ref{lemma_mult} we have
\begin{eqnarray*}
\id{\varPhi(a),\varPhi(b)}_{R^{nr}}&=&\sum_{i=1}^n\langle\rho(\alpha_i),\rho(\beta_i)\rangle_{R^{r}}=\sum_{i=1}^n\rho(1)M_{\alpha_i}\left(\rho(1)M_{\sigma(\beta_i)}\right)^{^T}\\
&=&\sum_{i=1}^n\rho(1)M_{\alpha_i}M_{\sigma(\beta_i)}^{^T}\rho(1)^{^T}\\
& =&\rho(1)\left(\sum_{i=1}^nM_{\alpha_i}M_{\sigma(\beta_i)}\right)\rho(1)^{^T}=0.
\end{eqnarray*}
This shows that $\varPhi(C^\perp)\subset\varPhi(C)^\perp$. We know $|C|=|\varPhi(C)|$. Since $A$ and $R$ are Frobenius rings, from Lemma \ref{lem_wood}, we have that
\[
|C^\perp|={|A|^n\over |C|}={|R|^{rn}\over |\varPhi(C)|}=|\varPhi(C)^\perp|.
\]
\end{proof}

\begin{example}
As in Example \ref{ex_symmain} assume $A=\F_2(\xi)[x]/(x^2)$ and $\sigma$ is the Frobenius map on $\F_4$ extended linearly to $A$. But, here assume $\mathcal{B}=(1,x+1)$. It is easy to check that $\mathcal{B}$ is symmetric but unlike the basis in Example \ref{ex_symmain}, $\sigma$ fixes $\mathcal{B}$ element-wise.
We again consider the code $C$ generated by
\[
\left(\begin{array}{cccc}
1&1+x\\
\end{array} \right)
\]
which is a $\sigma$-hermitian self-dual code over $A$. The code $\varPhi(C)$, whose generator matrix is,
\[
\left(\begin{array}{cccc}
1&0&0&1\\
0&1&1&0
\end{array}\right)
\]
is a $\sigma$-hermitian self-dual code over $\F_2(\xi)$. This was expected due to Theorem \ref{theo_sym}.
\end{example}

Using Magma, we determined that there are no pseudo-self-dual bases for $\F_2(\xi)[x]/(x^2)$ over $\F_2(\xi)$ nor $\F_2$. As we have seen there are however symmetric bases over $\F_2(\xi)$. The next Example shows the abundance of symmetric bases of $\F_2(\xi)$ over its various subrings.

\begin{example}
\label{F4[x]/x**2}
Assume $A=\F_2(\xi)[x]/(x^2)$ where $\xi^2+\xi+1=0$. $A$ has $18$ symmetric bases over $\F_2(\xi)$ (out of $90$ bases over $\F_2(\xi$)):
\[
\{1,x+1\},
\{\xi,\xi(x+1)\},
\{\xi^2,\xi^2(x+1)\}
\]
\[
\{1,\xi x+1\},
\{\xi,\xi^2x+\xi )\},
\{\xi^2,x+\xi^2)\}
\]
\[
\{1,\xi^2x+1\},
\{\xi,x+\xi )\},
\{\xi^2,\xi x+\xi^2)\}
\]
\[
\{x+1,\xi x+1\},
\{\xi(x+1),\xi^2x+\xi\},
\{\xi^2(x+1),x+\xi^2\}
\]
\[
\{x+1,\xi^2 x+1\},
\{\xi(x+1),x+\xi\},
\{\xi^2(x+1),\xi x+\xi^2\}
\]
\[
\{x+\xi,\xi^2 x+\xi\},
\{\xi x+\xi^2,x+\xi^2\},
\{\xi^2x+1,\xi x+1\}
\]
\begin{enumerate}
\item If $\sigma$ is the identity map, any of these symmetric bases will preserve duality. The inner product restricts to the Euclidean inner product on $\F_2(\xi)$.
\item If $\sigma$ is defined by $\xi \mapsto \xi^2$ and $x \mapsto x$, the unique symmetric basis that is fixed by $\sigma$ is $\{1,x+1\}$ meaning it is the only basis that preserves duality. The inner product restricts to the hermitian inner product on $\F_2(\xi)$.
\item If $\sigma$ is defined by $\xi \mapsto \xi^2$ and $x \mapsto \xi x$, the unique symmetric basis that is fixed by $\sigma$ is $\{1,\xi^2 x+1\}$ meaning it is the only basis that preserves duality. The inner product restricts to the hermitian inner product on $\F_2(\xi)$.
\item If $\sigma$ is defined by $\xi \mapsto \xi^2$ and $x \mapsto \xi^2 x$, the unique symmetric basis that is fixed by $\sigma$ is $\{1,\xi x+1\}$ meaning it is the only basis that preserves duality. The inner product restricts to the hermitian inner product on $\F_2(\xi)$.
\end{enumerate}

The ring $A$ has also $18$ symmetric bases over $\F_2$ (out of $840$ bases of $A$ over $\F_2$). Since only the identity map can fix the elements of a basis over the prime ring, for any involution of order 2, none of these bases are duality preserving.

The ring $A$ has also $3$ subrings isomorphic to $\F_2[x]/(x^2)$ (in some coding theory papers denoted as $\F_2+u\F_2$), one of them is $R=\F_2[x]/(x^2)\subset\F_4[x]/(x^2)$. There are $24$ symmetric bases of $A$ over $R$, one  such basis is  $\left(1,\xi\right)$ and 
$$M_1=\left(\begin{array}{cc}   1 & 0\\ 0 &1\end{array}\right) \mbox{ and }
M_\xi=\left(\begin{array}{cc} 0 &1\\1 &1\end{array}\right). $$
Another symmetric basis of $A$ over $R$ is
 $\left(    x + 1,
    \xi \right)$ and
$$M_{x+1}=\left(\begin{array}{cc}   x+1 & 0\\ 0 &x+1\end{array}\right) \mbox{ and }
M_\xi=\left(\begin{array}{cc} 0 &x+1\\x+1 &1\end{array}\right). $$
\end{example}

To finish this section, we study specifically $\F_q[x]/(x^t)$ which was considered in \cite{alfaro_2015}. Among other results, a condition on the change of basis matrix which changes from the standard basis, $\{1,x,\dots,x^{t-1}\}$, to some other basis is given which guarantees that the new basis preserves orthogonality. It turns out that this condition is equivalent to the new basis being symmetric. This is the subject of the next proposition.

\begin{proposition}
\label{prop_alf}
Assume $A=\F_q[x]/(x^r)$ and $R=\F_q$. Let $B$ be a change of basis matrix from the standard $\F_q$-basis for $A$ to the $\F_q$-basis $\mathcal{B}$. Then the following are equivalent:
\begin{enumerate}
\item $\mathcal{B}$ is symmetric.
\item $BB^{^T}$ is upper anti-triangular with constant anti-diagonal as well as all parallel diagonals (i.e.~for $b_{ij}=\left(BB^{^T}\right)_{ij}$, if $i+j>r+1$ then $b_{ij}=0$ and if $k=i+j\leq r+1$, $b_{1,k-1}=b_{2,k-2}\dots=b_{k-1,1}$).
\end{enumerate}
\end{proposition}
\begin{proof}
In the following we use the embedding of $A$ in $M_r(\F_q)$ as in Definition \ref{def_emb} where for $a\in A$, $M_a$ is its matrix representation and $\cdot$ is the Euclidean inner product on $\F_q^r$ with respect to $\mathcal{B}$. Then
$B^{^T}=\left[
        \rho(1)^{^T},
        \rho(x)^{^T},
        \ldots,
        \rho(x^{r-1})^{^T}
\right]$,
showing $\left(BB^{^T}\right)_{ij}=\rho(x^{i-1})\cdot\rho(x^{j-1})$.

Assume $\mathcal{B}$ is symmetric. Then for $m,n\geq 0$
\begin{eqnarray*}
\rho(x^m)\cdot\rho(x^n)&=&\rho(1)M_{x^m}\left(\rho(1)M_{x^n}\right)^{^T}=\rho(1)M_{x^m}M_{x^n}^{^T}\rho(1)^{^T}\\
&=&\rho(1)M_{x^m}M_{x^n}\rho(1)^{^T}=\rho(1)M_{1}M_{x^{m+n}}\rho(1)^{^T}\\
&=&\rho(1)M_{1}M_{x^{m+n}}^{^T}\rho(1)^{^T}=\rho(1)M_{1}\left(\rho(1)M_{x^{m+n}}\right)^{^T}\\
&=&\rho(1)\cdot\rho(x^{m+n}).
\end{eqnarray*}
With this it can easily be shown $BB^{^T}$ is upper anti-triangular with constant anti-diagonals.

Now assume $BB^{^T}$ is upper triangular with constant anti-diagonals. This condition is equivalent to saying  $\rho(1)\cdot\rho(x^k)=\rho(x)\cdot\rho(x^{k-1})=\dots=\rho(x^k)\cdot\rho(1)$. In some sense this says that the inner product respects the multiplication in $A$. Using this condition it can be shown that $\rho(ab)\cdot\rho(c)=\rho(a)\cdot\rho(bc)$ for $a,b,c\in A$. Let $a\in A$. Since $\left(M_aM_1^{^T}\right)_{ij}=\rho(v_ia)\cdot\rho(v_j)=\rho(v_i)\cdot\rho(v_ja)=\left(M_1M_a^{^T}\right)_{ij}$
we have that
$
M_a=M_aM_1=M_aM_1^{^T}=M_1M_a^{^T}=M_a^{^T}$.
Hence, $\mathcal{B}$ is symmetric.
\end{proof}

Now that we see that the condition presented in \cite{alfaro_2015} Proposition 5 is equivalent to a basis being symmetric, it is clear that their condition guarantees that the basis preserves duality which is the essence of Proposition 5 and Corollary 1 in \cite{alfaro_2015}.

\section{Examples using cyclic self-dual codes}
\label{sect_examples}
We give some examples of rings and bases which show that the best minimum Gray weight obtained for a family of self-dual codes depends of the choice of the basis for which the Gray weight is with respect to. From this it can be concluded that, for a given problem, not all bases share the same properties and that some care has to be taken when choosing a basis.

A cyclic code $C$ over a finite commutative ring $A$ is an ideal $(g)/(Y^n-1)\subset A[Y]/(Y^n-1)$, where the polynomial $g=g_0+\ldots+g_{n-k-1}Y^{n-k-1}+g_{n-k}Y^{n-k}$ is a divisor of $Y^n-1\in A[Y]$. The generating matrix of the corresponding linear code is
{\small $$G= \left(\begin{array}{ccccccc}
g_0& \ldots & g_{n-k-1} &  g_{n-k}  & 0 &\ldots &0  \\
0  &g_0& \ldots & g_{n-k-1} &  g_{n-k}  & \ldots&0 \\
0  & \ddots &\ddots  & \ddots&  \ddots& \ddots&\vdots \\
0  &  &  & &  && \\
0 & \ldots &0&
g_0& \ldots & g_{n-k-1} &  g_{n-k}
\end{array}\right).
$$}

If $h_0\in A$ is invertible, then the reciprocal polynomial of $h=\sum_{i=0}^r h_iY^i\in A[Y]$ is  $h^*=\frac{1}{h_0}\sum_{i=0}^r h_{r-i}Y^i$. For an automorphism $\sigma$ of order $2$ we also define $h^*_{\sigma}=\frac{1}{\sigma(h_0)}\sum_{i=0}^r \sigma(h_{r-i})Y^i$. It is well known that the cyclic code $(g)/(Y^n-1)\subset A[Y]/(Y^n-1)$ is euclidian self-dual if and only if $Y^n-1=h\cdot g$ and $g=h^*$ and is $\sigma$-hermitian self-dual if and only if $Y^n-1=h\cdot g$ and  $g=h^*_{\sigma}$.

Row $i+1$ of the above generator matrix of $C$ corresponds to the polynomials $Y^i g$. If $A$ is a free $R$-module for some subring $R$ of $A$, then for an $R$-basis of $A$, $\{v_1,\dots,v_r\}$, $C$ is spanned over $R$ by
\[
v_1g,\cdots,v_rg,v_1Yg,\cdots,v_rYg,\cdots,v_1Y^kg,\cdots,v_rY^kg.
\]
Therefore the lines of a generator matrix of $\varPhi(C)$ are given by the image under $\varPhi$ of the vectors corresponding to these polynomials.

\subsection{Euclidean self-dual cyclic codes over $\F_{4}[x]/(x^3)$}
Denote $A=\F_{4}[x]/(x^3)$ and $\F_4=\F_2(\xi)$. We obtain the following decompositions in $A[Y]$ (which is not a unique factorization domain)
\begin{eqnarray*}
Y^2-1 & = & (Y + 1)(Y + 1)\\
 & = & (Y + x^2 + 1)(Y + x^2 + 1)\\
& = & (Y + \xi x^2 + 1)(Y + \xi x^2 + 1)\\
& = & (Y + \xi^2x^2 + 1)(Y + \xi^2x^2 + 1)
\end{eqnarray*}
where each constant term is its own inverse. Therefore there exists four self-dual cyclic codes $(g)/(Y^2-1)\subset A[Y]/(Y^2-1)$ over $A$, generated by $g_1=Y + 1$, $g_2=Y + x^2 + 1$, $g_3=Y + \xi x^2 + 1$ and $g_4=Y + \xi^2x^2 + 1$. In order to map those self-dual codes over $A$ to self dual codes over $\F_4$ we verified the symmetric basis criteria for all $30240$ $\F_4$-bases of $A$ and found $480$ distinct symmetric bases. Also, there are no (pseudo)-self-dual bases over $\F_4$. We consider the following two symmetric bases of $A$ over $R=\F_4$.
$${\cal B}_1=[
    \xi^2x + 1,\,
    \xi x,\,
    \xi^2x^2 + \xi^2x + 1
];\; {\cal B}_2= [
    \xi^2x^2 + \xi^2,
    x^2 + \xi^2x + \xi,
    \xi^2x^2 + x + 1
]
$$
\begin{enumerate}
\item Using the basis ${\cal B}_1$ we obtain that the image codes over $\F_4$ of the self dual codes $(g)/(Y^2-1)\subset A[Y]/(Y^2-1)$ over $A$ has the following generator matrix (Here the lines of the generator matrix corresponds to the coefficients of $(\xi^2x + 1)\cdot g,(\xi x)\cdot g$ and $(\xi^2x^2 + \xi^2x + 1)\cdot g$ in the basis ${\cal B}_1$) and  weight enumerator:
\begin{itemize}
\item $g_1$:
$\left(\begin{array}{cccccc}
1&0&0&1&0&0\\
0&1&0&0&1&0\\
0&0&1&0&0&1
\end{array}\right)$  and $1+9w^2+27w^4+27w^6$.
\item $g_2$:
$\left(\begin{array}{cccccc}
\xi^2&0&\xi&1&0&0\\
0&1&0&0&1&0\\
\xi&0& \xi^2&0&0&1
\end{array}\right)$ and $1+3w^2+12w^3+3w^4+36w^5+9w^6$
\item $g_3$:
$\left(\begin{array}{cccccc}
\xi&0 &\xi^2&1&0&0\\
0&1&0&0&1&0\\
\xi^2&0&\xi&0&0&1
 \end{array}\right)$ and $1+3w^2+12w^3+3w^4+36w^5+9w^6$
\item $g_4$:
$\left(\begin{array}{cccccc}
0&0&1&1&0&0\\
0&1&0&0&1&0\\
1&0&0&0&0&1
\end{array}\right)$ and $1+9w^2+27w^4+27w^6$
\end{itemize}
\item Using the basis ${\cal B}_2$ we obtain that the image codes over $\F_4$ of the self dual codes $(g)/(Y^2-1)\subset A[Y]/(Y^2-1)$ over $A$ has the following generator matrix and  weight enumerator:
\begin{itemize}
\item  $g_1$ : we obtain the same code   $\varPhi(C)$ as for ${\cal B}_1$.
\item  $g_2$:
$\left(\begin{array}{cccccc}
\xi&\xi&1&1&0&0\\
\xi&0 &\xi^2&0&1&0\\
1 &\xi^2& \xi^2&0&0&1
\end{array}\right)$ and $1+6w^3+27w^4+18w^5+12w^6$
\item $g_3$:
$\left(\begin{array}{cccccc}
0& \xi^2&\xi&1&0&0\\
\xi^2 &\xi^2&1&0&1&0\\
\xi&1&\xi&0&0&1
 \end{array}\right)$ and $1+6w^3+27w^4+18w^5+12w^6$
\item $g_4$:
$\left(\begin{array}{cccccc}
\xi^2&1& \xi^2&1&0&0\\
1&\xi&\xi&0&1&0\\
\xi^2&\xi&0&0&0&1
\end{array}\right)$ and $1+6w^3+27w^4+18w^5+12w^6$
\end{itemize}
\end{enumerate}
Using the first basis we would not have obtained a self-dual code over $\F_4$ of optimal hamming weight $3$ (cf. \cite{gaborit_2003}).

\subsection{Hermitian self-dual cyclic codes over $\F_9[x]/(x^2-2)$}
Consider the ring $A=\F_3(\xi)[x]/(x^2+1)$ where $\xi^2+2\xi+2=0$ from Example \ref{ex_main}. We obtain two factorizations of $Y^2-1$ in $A[Y]$ (which is not a unique factorization domain since $A\cong\F_9\oplus\F_9$):
\[
Y^2-1=(Y + \xi^6x)(Y + \xi^2x)=(Y+1)(Y+2).
\]

Each factor produces a $[2,1]$ cyclic code over $A$ none of which are Euclidian self-dual codes nor $\theta^2$-hermitian self-dual codes. However, the codes generated by the factors $Y+\xi^2x$ and $Y+\xi^6x$ generate codes that are both $\psi$-hermitian self-dual and $\psi\theta^2$-hermitian self-dual.
\begin{itemize}
\item From the table in Example \ref{ex_main} we see there are $64$ $\psi$-pseudo-self-dual bases with respect to the subgroup $\id{\psi}$. Since $\id{\psi}$ fixes $\F_3(\xi)$ element-wise, each one of the bases will preserve duality. So, by Theorem \ref{theo_tob}, the $\F_3(\xi)$ images of the codes generated by $Y+\xi^2x$ and $Y+\xi^6x$ will each be a hermitian self-dual code over $\F_3(\xi)$. The best hamming distance of any of these images is 2.
\item From the table in Example \ref{ex_main} we see there are $96$ $\theta^2\psi$-pseudo-self-dual bases with respect to the subgroup $\id{\psi}$. Since $\id{\psi}$ fixes $\F_3(\xi)$ element-wise, each one of the bases will preserve duality. So, by Theorem \ref{theo_tob}, the $\F_3(\xi)$ images of the codes generated by $g_1=Y+\xi^2x$ and $g_2=Y+\xi^6x$ will each be a hermitian self-dual code over $\F_3(\xi)$.
\begin{enumerate}
\item For the basis ${\cal B}_1=(\xi x,\xi^6)$ the generator matrix of the image of this codes are respectively
\[
\left(\begin{array}{cccc}
  0  & \xi&   1  & 0\\
\xi^7  & 0 &  0   &1
\end{array} \right) \mbox{ and }
\left(\begin{array}{cccc}
  0 &\xi^5 &  1  & 0\\
\xi^3  & 0  & 0 &  1
\end{array} \right) \]
which are both of hamming distance $2$.
\item For the basis ${\cal B}_2=(2x+\xi^3,\xi^6x+\xi)$ the generator matrix of the image of this codes are respectively
\[
\left(\begin{array}{cccc}
\xi^2 &\xi^2  & 1  & 0\\
\xi^2 &\xi^6 &  0  & 1
\end{array} \right)  \mbox{ and }
\left(\begin{array}{cccc}
\xi^6 &\xi^6 &  1  & 0\\
\xi^6 & \xi^2  & 0  & 1
\end{array} \right) \]
which are both of hamming distance $3$.
\end{enumerate}
\end{itemize}
This illustrates that, the properties of all image codes of some family of self-dual codes $C$ over $A$, can depend on the choice of the $\sigma$-pseudo-self-dual basis.

\section*{Acknowledgements}
The first author would like to thank IRMAR for the support of his visit to their institution during which time this work was initiated.

\appendix
\section{Pseudo-Self-Dual Bases of $M_2(\F_2)$}
\label{ap_M2F2}
\begin{example}
\label{m22b}
Let $A=M_2(\F_2)$. In \cite{bachoc_1997}, it was shown that codes over $A$ can be mapped to codes over $\F_4$. In \cite{alahmadi_2013} self-dual cyclic codes over $A$ were studied. In the process they show (Proposition 2 of \cite{alahmadi_2013}) that there is a map from $A^n$ to $\F_4^{2n}$ that preserves self-orthogonality for a particular hermitian form. It turns out that this map is defined using a $\sigma$-self-dual $\F_4$ basis for $A$ and the hermitian form is the standard $\sigma$-hermitian form on $A$ where $\sigma$ is the anti-transpose on $A$. Much more can be said about this map, specifically that it preserves duality for not only the hermitian form based on the anti-transpose on $A$ but also the hermitian form based on the transpose on $A$ and a few others. This is what we shall show here. The following identifications for the elements of $M_2(\F_2)$ will be used throughout.
\[
I=\left[\begin{array}{cc}1 &0 \\0 & 1 \end{array} \right],\;
z=\left[\begin{array}{cc}0 &0 \\0 & 0 \end{array} \right],\;
i=\left[\begin{array}{cc}0 &1 \\1 & 0 \end{array} \right],\;
a=\left[\begin{array}{cc}1 &1 \\1 & 1 \end{array} \right]
\]
\[
e_1=\left[\begin{array}{cc}1 &0 \\0 & 0 \end{array} \right],\;
e_2=\left[\begin{array}{cc}0 &1 \\0 & 0 \end{array} \right],\;
e_3=\left[\begin{array}{cc}0 &0 \\1 & 0 \end{array} \right],\;
e_4=\left[\begin{array}{cc}0 &0 \\0 & 1 \end{array} \right]
\]
\[
u_1=\left[\begin{array}{cc}0 &1 \\1 & 1 \end{array} \right],\;
u_2=\left[\begin{array}{cc}1 &0 \\1 & 1 \end{array} \right],\;
u_3=\left[\begin{array}{cc}1 &1 \\0 & 1 \end{array} \right],\;
u_4=\left[\begin{array}{cc}1 &1 \\1 & 0 \end{array} \right]
\]
\[
t=\left[\begin{array}{cc}1 &1\\0 & 0 \end{array} \right],\;
b=\left[\begin{array}{cc}0 &0 \\1 & 1 \end{array} \right],\;
l=\left[\begin{array}{cc}1 &0 \\1 & 0 \end{array} \right],\;
r=\left[\begin{array}{cc}0 &1 \\0 & 1 \end{array} \right]
\]

We first compute the automorphisms and anti-automorphisms of $A$. The set of units of $R$ is $\{I,i,u_1,u_2,u_3,u_4\}$. The units $i$, $u_2$ and $u_3$ are of multiplicative order 2 and the units $u_1$ and $u_4$ are of multiplicative order 3. The set $\{I,i,u_1,u_2\}$ is an $F_2$-basis for $A$. Replacing $i$ or $u_2$ with $u_3$ will still be a basis, as will replacing $u_1$ with $u_4$. Now, any automorphism or anti-automorphism of $A$ must send units to units of the same order. There are 12 such maps $A\to A$ that send units to units of the same order. They form a group isomorphic to a subgroup of $S_6$. Consider three of these maps which we express in cycle notation on the set of units. Let $\tau$ be the map $(u_2u_3)$, $\psi$ be the map$(u_1u_5)$ and $\theta$ be the map $(iu_2u_3)$. It turns out that $\tau$ is the transpose, $\psi$ is the anti-transpose (the reflection of a matrix about the anti-diagonal) and
\[
\theta:\left[\begin{array}{cc}a&b\\c&d\\\end{array}\right]\mapsto \left[\begin{array}{cc}b+d&a+b+c+d\\b&a+b\\\end{array}\right].
\]
Note that $\tau$ and $\psi$ are anti-automorphisms on $A$ and $\theta$ is an automorphism on $A$. From this we deduce that the automorphism group is isomorphic to $S_3$ and generated by $\tau\psi$ and $\theta$ which is $\{id,\theta,\theta^2,\tau\psi,\tau\psi\theta,\tau\psi\theta^2\}$. The set of anti-isomorphisms are $\{\tau,\psi,\tau\theta,\psi\theta,\tau\theta^2,\psi\theta^2\}$. The subset of these, $\{\tau,\psi,\tau\theta,\tau\theta^2\}$, is the set of involutions of $A$, the remaining anti-isomorphisms having order 6.

Let $R=\left\{I,z,u_1,u_4\right\}$. Notice, $R$ is a subring of A isomorphic to $\F_4$ and $A$ is a free left $R$-module with left $R$-basis $\mathcal{B}=(I,i)$. Notice $\psi$ is simply the Frobenius map when restricted to $R$ and is the conjugation map on $A$ from Section 4.2 in \cite{alahmadi_2013}. Of course $\psi(R)=R$. Let  $H$ be the group generated by $\theta$ whose fixed ring is $R$. Since $Tr_{H}(I\psi(i))=Tr_{H}(i\psi(I))=z$ and $Tr_{H}(I\psi(I))=Tr_{H}(i\psi(i))=I$, $\mathcal{B}$ is a $\psi$-self-dual basis w.r.t. $H$. Since $A$ and $R$ are Frobenius, by Theorem \ref{theo_tob}, the dual of a code over $A$ is mapped to the dual of the image of the code over $R$ where the hermitian form being considered\ is based on the anti-transpose. Since $\psi$ restricted to $R$ is simply the Frobenius map on $\F_4$ we are considering standard $\psi$-hermitian form on $R$.

It can be similarly shown that $\mathcal{B}$ is a $\tau$-self-dual basis w.r.t. $H$. The difference here is that $\tau$ is the identity on $R$. So, the $\tau$-hermitian form restricted to $R$ is simply standard Euclidean form. Using Magma, we found all $\sigma$-pseudo-self-dual $\F_4$-bases and $\F_2$-bases for $M_2(\F_2)$ for each given involution, $\sigma$.
\begin{enumerate}
\item These are the $\sigma$-pseudo-self-dual $\F_4$-bases w.r.t. $H=\id{\theta}$ for $M_2(\F_2)$ for each involution $\sigma$:
\begin{itemize}
\item $\sigma=\psi$: $
\{I,i\},\{I,u_2\},\{I,u_3\},\{i,u_1\},\{i,u_4\},\{u_1,u_2\},\{u_1,u_3\}$, $\{u_2,u_4\}$, $\{u_3,u_4\}$
\item $\sigma=\tau$ :
$
\{I,i\},
\{e_1,e_2\},
\{l,r\},
\{e_3,e_4\},
\{u_1,u_2\},
\{u_3,u_4\}
$.
\item $\sigma=\tau\theta$:  $
\{I,u_2\},
\{i,t\},
\{u_1,u_3\},
\{r,a\},
\{b,e_4\},
\{e_2,t\}
$
\item
$\sigma=\tau\theta^2$: $
\{I,i\},
\{u_1,u_2\},
\{u_3,u_4\},
\{e_1,e_2\},
\{e_3,e_4\}
$
\end{itemize}
\item
 The following are the $\sigma$-pseudo-self-dual $\F_2$-bases w.r.t. $H=Aut(A)$ for $A$ for each involution $\sigma$ :
 \begin{itemize}
\item
$\sigma=\psi$: none\\
$\sigma=\tau$:
$\{u_1,u_2,u_3,u_4\},
\{e_1,e_2,e_3,e_4\}$
\item
$\sigma=\tau\theta$ : $
\{i,u_1,u_3,u_4\},
\{r,a,e_2,t\} $
\item $\sigma=\tau\theta^2$:
$\{i,u_1,u_2,u_4\},
\{r,a,b,e_3\}$
\end{itemize}
\end{enumerate}
\end{example}

\section{Examples of Symmetric Bases}

\begin{example} The Galois ring $A=GR(4,2)$ defined in example \ref{GR} has $24$ symmetric bases (among which are  the $8$  $\sigma$-pseudo-self-dual bases previously found). An example of a new basis is   $\{1,\xi+2\}$.
\end{example}
\begin{example} The ring $\Z_4[x]/( x^2-2x)$ (see \cite{martinez_2015_3}) has $16$ symmetric basis.  An example of such a  basis is   $\{1,x+1\}$.
\end{example}
\begin{example} The ring $\Z_4[x]/( x^2-2)$  has $16$ symmetric bases.  An example of such a  basis is   $\{1,x+1\}$.
\end{example}
\begin{example} The ring $\Z_4[x]/(x^2-x)$ (\cite{gao_2014}) has no $\sigma$-pseudo-self-dual bases, but has $8$ symmetric bases over $\Z_4$.  An example of such a  basis is $\{x+3,x\}$.
\end{example}

\begin{example} The ring $\F_2[x]/( x^4)$  has no $\sigma$-pseudo-self-dual bases, but has $12$ symmetric bases over $\F_2$.  An example of such a  basis is   $\{
        x^2 + x,
        x^3 + x,
        1,
        x^3 + 1
    \}$.
\end{example}
 \begin{example} The ring $\F_2[x]/( x^2+y^2, xy)$
 (see \cite{martinez_2015})
 has $16$ symmetric bases over $\F_2$.  An example of such a  basis is   $\left(
       y^2 + 1,
        y,
        x,
        1
    \right)$.
 \end{example}

  \begin{example} The ring $A=\F_2[x]/( x^2,y^2)$   has $8$ symmetric basis over $\F_2$.  An example of such a  basis is   $\left(
        y + 1,
        x + 1,
        xy + x + y + 1,
        1
    \right)$. The ring $A$ has also $7$ subrings isomorphic to $\F_2[x]/( x^2)$. For the subring $R=\{0,1,x,x+1\}$ there are $48$ symmetric bases of $A$ over $R$. One  such a basis is  $\left(  y + 1,
        x + 1\right)$ for which
$$M_{y+1}=\left(\begin{array}{cc}   0 & x+1\\ x+1 &0\end{array}\right) \mbox{ and }
M_{x+1}=\left(\begin{array}{cc}  x+1 &0\\0 & x+1\end{array}\right). $$
 \end{example}
  \begin{example} The ring $\F_2[x]/( x^3)$  of order $8$ has no self dual basis but has four  symmetric bases:
 $\left(
        x^2 + x,
        x^2 + 1,
        1
 \right)$,
 $\left(
        x^2 + 1,
        x,
        1
 \right)$,
 $\left(
        x^2 + x,
        x + 1,
        x^2 + x + 1
 \right)$,
 $\left(
        x,
        x + 1,
        x^2 + x + 1
 \right)$
  \end{example}
     \begin{example} The ring $\F_2[x]/( x^3-1)$  of order $8$ has no self dual basis but has three  symmetric bases:
 $\left(
        x^2 + 1,
        x + 1,
        x^2 + x + 1
 \right)$,
 $\left(
        x^2 + x,
        x + 1,
        x^2 + x + 1
\right)$, \\
 $\left(
        x^2 + x,
        x^2 + 1,
        x^2 + x + 1
 \right)$.
   \end{example}
\begin{example} The self dual basis $ \left(
        x,
        x + 1
 \right)$ of  the ring $\F_2[x]/( x^2 +x)$ of order $4$ is also a symmetric basis.
  \end{example}
   \begin{example} The ring $\F_2[x]/( x^2)$  of order $4$ (cf. (\cite{dougherty_1999_3})) has no $\sigma$-pseudo-self-dual basis but has the unique  symmetric basis    $\left(
              x + 1,
        1
    \right)$.
 \end{example}

\bibliographystyle{amsplain}
\bibliography{../SteveSzaboRefs}

\end{document}